\documentclass[12pt]{article}
\usepackage[T1]{fontenc}
\usepackage{kpfonts}
\usepackage[dvipsnames]{xcolor}
\usepackage{pdfpages}
\usepackage{sgame}
\usepackage{accents}
\usepackage{tikz-cd}
\usepackage{float}
\usepackage[round]{natbib}
\usepackage{multirow}
\usepackage{dutchcal}
\usepackage{pdfpages}
\usepackage{bbm}
\usepackage{sgame}
\usepackage{mathtools}
\usepackage{amsthm}
\usepackage{enumitem}
\usepackage[margin=1in]{geometry}
\usepackage{caption}
\usepackage{subcaption}
\usepackage{epigraph}
\usepackage{listings}
\usetikzlibrary{calc}
\usetikzlibrary{shapes,arrows}
\usepackage{tcolorbox}

\newcommand{\ubar}[1]{\underaccent{\bar}{#1}}
\DeclareMathOperator\supp{supp}

\DeclareMathOperator*{\argmax}{arg\,max}

\newtheorem{theorem}{Theorem}[section]
\newtheorem{proposition}[theorem]{Proposition}
\newtheorem{lemma}[theorem]{Lemma}
\newtheorem{corollary}[theorem]{Corollary}

\theoremstyle{definition}

\usepackage{hyperref}
\hypersetup{
    colorlinks=true,
    linkcolor=Periwinkle,
    filecolor=Periwinkle,      
    urlcolor=Periwinkle,
    citecolor=Periwinkle,
}

\definecolor{backcolour}{rgb}{0.63, 0.79, 0.95}

\lstdefinestyle{mystyle}{
  backgroundcolor=\color{backcolour},
  basicstyle=\ttfamily\footnotesize,
  breakatwhitespace=false,         
  breaklines=true,                 
  captionpos=b,                    
  keepspaces=true,                 
  numbers=left,                    
  numbersep=5pt,                  
  showspaces=false,                
  showstringspaces=false,
  showtabs=false,                  
  tabsize=2
}

\providecommand{\keywords}[1]{\textbf{\textit{Keywords:}} #1}
\providecommand{\jel}[1]{\textbf{\textit{JEL Classifications:}} #1}

\begin{document}
\title{Costly Evidence and Discretionary Disclosure}
\author{Mark Whitmeyer\thanks{Arizona State University. Email: \href{mailto:mark.whitmeyer@gmail.com}{mark.whitmeyer@gmail.com}} \and Kun Zhang\thanks{Arizona State University. Email: \href{kunzhang@asu.edu}{kunzhang@asu.edu}.}}

\date{\today}

\maketitle

\begin{abstract}
A sender flexibly acquires evidence--which she may pay a third party to certify--to disclose to a receiver. When evidence acquisition is overt, the receiver observes the evidence gathering process irrespective of whether its outcome is certified. When acquisition is covert, the receiver does not. In contrast to the case with exogenous evidence, the receiver prefers a strictly positive certification cost. As acquisition costs vanish, equilibria converge to the Pareto-worst free-learning equilibrium. The receiver always prefers covert to overt evidence acquisition.
\end{abstract}
\\
\jel{C72; C73; D82; D83}

\newpage

\section{Introduction}

This paper studies a disclosure game in which evidence is endogenously and flexibly acquired at a cost by a sender. Once she has acquired the evidence, the sender chooses whether to pay a third-party firm to certify it.\footnote{Equivalently, the sender must incur some cost in order to disclose her findings.} In the absence of such certification, the sender's message is merely cheap talk (equivalently, the sender simply does not disclose her evidence). Upon observing the sender's message the receiver takes an action that affects both players' payoffs.

Real-world examples that are captured by this setup abound. For instance, this model is an extension of the seminal papers \cite{verr1} and \cite{verr2}, in which the sender is a manager of a risky asset and the receiver is a collection of market traders. It is natural to suppose that the evidence available to the manager for disclosure takes time and effort to acquire. Another suitable example is that of a monopolist generating information about its product in order to convince consumers to purchase it. Product tests are costly yet can be designed with care, justifying the flexible approach we allow our sender.

We grant the sender complete flexibility over what evidence she acquires, but impose that more and finer information is more costly to acquire. This is realistic. Consider our leading example, in which the sender is the manager of a firm (or the firm itself) and the state is the value of some economic asset. In this setting, the manager has significant freedom about the type and quality of information she acquires. She can instruct her employees to scrutinize certain aspects of the asset more than others, hire accountants with expertise in different areas, or conduct tests or reviews with different emphases. Furthermore, more (finer) information is more expensive. More thorough tests are more expensive to run, more skilled employees are more expensive to hire, and hiring more workers costs more.

Within this environment we are interested in the following three things. First, what is the impact of transparency in the evidence gathering stage? We look at two different cases--covert, in which the sender's evidence acquisition is secret; and overt, in which it is public--and compare the equilibria between the two. Second, what is the role played by evidence acquisition costs? In particular, how do the equilibria and players' welfare change as it becomes more difficult for the sender to acquire information? As such costs vanish, what equilibrium in the costless acquisition benchmark serves as the limit? Third, what is the role of the certification fee given that the evidence is exogenous?

We find that the receiver always prefers covert to overt evidence acquisition, strictly if there is an equilibrium with covert acquisition in which the sender acquires information. That is, transparency is always detrimental to the receiver's welfare. This phenomenon arises  since transparency grants the sender a form of partial commitment absent when her evidence acquisition is hidden. Our sender would like to credibly commit not to acquire any information at all, but when evidence acquisition is covert and information is not too costly, she finds it too tempting to secretly deviate and acquire information, to the receiver's benefit.

That the overtness of evidence acquisition may lead to total uninformativeness is also observed by \cite{demarzo} (who specify that a sender's test may fail to return evidence \'{a} la \cite{dye}), and the intuition behind their result is identical to ours. They also find that the receiver may strictly benefit from opacity but the intuition behind their result is different from that for ours. Specifically, in their setting, all tests are equally costly (or costless), and so a ``minimum principle'' holds, which imposes that the market is maximally pessimistic about the asset's quality following non-disclosure. Not so in our setting; instead, the belief by our receiver (or the market) upon observing non-disclosure is the fixed-point lower point of support in our sender's optimal flexible learning.

Even when evidence acquisition is covert, if the costs of obtaining evidence or certification are too high, the sender acquires no information in any equilibrium. However, if obtaining evidence is sufficiently cheap and the certification cost is relatively low, the sender acquires evidence in every equilibrium. We also look at what happens as information costs vanish. We find that when this occurs, all equilibria converge to the Pareto-worst free-learning equilibrium, which highlights a significant difference between cheap information and free information. This is reminiscent of the main finding of \cite{ravid2020learning}, but the mechanics behind the results are different. Our result occurs because when information is costly, the sender always seeks extremely coarse evidence: she either acquires evidence that the state is high, which she gets certified; or evidence that it is low, which she does not.

We also look at the effect of the certification fee. Perhaps surprisingly, since it stands in stark contrast to the findings in the papers with exogenous information, when evidence acquisition is covert, the receiver always weakly prefers a strictly positive certification fee and may strictly prefer it. The certification fee cannot be too high, or else the sender will not learn, but a low but positive certification fee strictly improves the receiver's welfare.

The reason for this result is subtle and stems from the endogeneity (and cost) of the information available to the sender. There cannot be an equilibrium in which the sender acquires evidence when certification is free since the receiver's conjectured belief following non-certification \textit{can never be correct}. Because certification is free, the sender's payoff as a function of the expected value of the state is locally strictly convex at the conjectured belief (there is a kink precisely at the conjectured expected value) and so it is never optimal for the sender to obtain that value and not certify it.\footnote{We illustrate this in Figure \ref{fig3} \textit{infra}.} Given this, the only possible equilibrium is that the sender acquires no evidence and always gets it certified, which is indeed the equilibrium when there is no certification cost.

We finish this section by discussing related work. Section \ref{model} introduces the model, Section \ref{bench} goes through a benchmark when evidence acquisition is costless, and Section \ref{main} contains the (bulk of the) main results. Section \ref{pareto} details our free-learning equilibrium selection result, and Section \ref{discus} concludes. Appendix \ref{eqdef} specifies our equilibrium concept, and Appendix \ref{omitproof} contains the proofs omitted from the text. We argue that our findings persist when our setting is instead that of \cite{dye} in a \href{https://whitmeyerhome.files.wordpress.com/2022/08/costly_evidence_and_discretionary_disclosure_supplement_v2.pdf}{Supplementary Appendix}.

\subsection{Related Work}

\cite{shishkin}, who also endogenizes evidence acquisition in the model of \cite{dye} but imposes that such acquisition is costless, is the closest paper to this one. \citeauthor{shishkin} also distinguishes between the overt and covert scenarios and shows that the equilibrium evidence structures take remarkably simple forms: for a low failure rate it is a pass-fail test; and it is two-sided censorship (in which an intermediate region of states is fully disclosed) otherwise. Our sender also chooses simple evidence (she always acquires at most two pieces of evidence), but in our setting this is a consequence of the fact that more precise evidence is more expensive.

\cite{transwhitjain} is also similar. There, the authors study a disclosure game with two senders, each of whom acquires evidence (for free) in order to obtain the favor of a receiver. The commitment power engendered by overt evidence acquisition (what they term ``transparency of protocol'') also makes the receiver worse off vis-a-vis the covert scenario, though the mechanism is different to this one. Things are more complicated there since a sender's payoff is itself an equilibrium object, generated by the strategy of the other sender.

This paper is also related to \cite{bert}, who study a problem in which a regulator designs a measurement system (a signal) about an asset's value, the realization of which is observed by manager (privately) with some exogenous probability, which she may then disclose to investors (the receiver) if she obtained evidence. Their problem is similar to \cite{shishkin} in that generating evidence is costless, although in their setting it is a third party who designs the information. 

\cite{demarzo} also endogenizes evidence acquisition in the setting of \cite{dye} by allowing a sender to choose a test from set $K$ of available tests, before choosing whether to disclose its outcome or pretend to have received a null result (which occurs with some probability). In part of the paper they assume that tests are costly but importantly this cost is \textit{independent of the test, its result, or the asset's quality (the state)}, which is not the case in our paper. This creates significant divergences between our results and theirs: they find that as the testing costs vanish, nearly all of their agents (the testing cost, for them, is a random variable) take the test; whereas our limiting outcome, in contrast, is particularly inefficient: our sender chooses the Pareto-worst equilibrium evidence structure. Likewise, our finding about the benefits (to the receiver) of a certification cost (or positive test failure rate) is absent from their paper since their constant cost precludes the local convexity at the conjectured belief property that we encounter.

Other works that endogenize the sender's information in disclosure games are \cite{pae}, who, like \cite{demarzo}, allows the sender to pay an nondiscriminatory cost to acquire information about the (in \citeauthor{pae}'s paper) endogenous state; and \cite{hughespae}, who stipulate that the sender may pay some (again, fixed) cost to receive the ``precision information'' about the unknown state. In \cite{bendekellip}, the manager (sender) chooses from different projects the random outcome from which she may disclose. Surprisingly, the manager takes excessive risks even if the expected return is lower. Competition is a key feature of \cite{kartikleesuen} who study a competitive disclosure game with multiple senders. Each (privately) chooses the probability with which she obtains evidence (subject to an increasing cost), which she may subsequently disclose. \cite{hardinfo} study the problem of a buyer who chooses publicly what (distribution over) hard evidence to obtain, which she may then disclose to a monopolistic seller who engages in nonlinear pricing.\footnote{A couple other less similar but still relevant papers are \cite{henry}, who also studies transparency in a disclosure game with information acquisition in which acquired evidence can be concealed; and \cite{penno}, who allows a sender to choose the precision of the firm's financial reporting system.}

There is also a literature on information acquisition in cheap-talk games. In a general cheap-talk environment, \cite{pei} looks at covert costly information acquisition and establishes that a sender always communicates what she finds to the receiver. The intuition driving this result is similar to that implying our Lemmas \ref{mainlemma1} and \ref{mainlemma2}. \cite{argenz} compare covert and overt costly information acquisition, and find that under both regimes, the sender over-acquires information and that the receiver may prefer either to delegation. \cite{escude} allows a sender to costlessly and covertly acquire partially verifiable information about the state, which he discovers is always revealed by the sender. In \cite{lyusuen}, the authors study overt, costless information acquisition preceding cheap talk.

Finally, our sender, taking as given her disclosure strategy, faces a Bayesian persuasion/information acquisition problem (\cite{kam} and \cite{kamcostly}), and so our paper fits broadly within these literatures.\footnote{\cite{macksurvey} and \cite{bpsurvey} provide comprehensive surveys of the two areas.} Specifically, we use techniques from \cite{gent}, \cite{linearprog}, \cite{martini} and \cite{bipool} throughout. 

To the best of our knowledge, we are the only paper to study costly flexible evidence acquisition in a game of verifiable disclosure. This is valuable for two reasons. First, it is realistic: an asset manager has significant freedom about what evidence to acquire, and finer and more granular information is more expensive. Second, it leads to effects and predictions absent from settings with exogenous evidence, endogenous but free (or equally costly) evidence, or endogenous evidence from a family of distributions parameterized by a scalar.

\section{The Model}\label{model}

There is a sender (she, $S$) and a receiver (he, $R$). There is an unknown state of the world $\theta \in \left[0,1\right]$ about which $S$ and $R$ share a common prior with cdf $F$ and mean $\mu \in \left(0,1\right)$. $F$ admits a strictly positive density $f$. For technical convenience, $F$ is strictly log-concave. $R$'s set of actions is $A = \left[0,1\right]$ and her utility is the commonly-used ``quadratic'' loss specification: $u_{R}\left(a,\theta\right) = -\left(a-\theta\right)^2$. The sender's utility is state-independent and linear in the sender's action: $u_{S} = a$.

Before communicating to $R$, $S$ acquires information flexibly subject to a cost. We allow her considerable freedom in this regard: she may choose any signal $\pi \colon \left[0,1\right] \to \Delta\left(X\right)$, where $X$ is a (compact) set of signal realizations.\footnote{For a compact metrizable space \(Y\), let \(\Delta (Y)\) denote the set of all probability measures on the Borel subsets of \(Y\).} As is well-known, the quadratic loss utility of $R$ means that $R$'s (uniquely) optimal action at any posterior distribution is simply the posterior expectation.\footnote{Thus, an alternative interpretation of our setup is that given by \cite{verr1} and the literature that follows: he specifies that the receiver is a market and that the sender wants to maximize the expected market value of the project (belief about the state).} For this reason, it is without loss to restrict attention on the class of signals where \(X = [0,1]\), and each \(x \in X\) equals the induced posterior mean: \(x = \mathbb{E}[\theta | x]\). Moreover, the set of distributions of posterior means that is consistent with some signal is the set of mean-preserving contractions (MPCs) of the prior $F$,\footnote{This is established formally in \cite{linearprog}. A distribution \(G \in \Delta([0,1])\) is a mean-preserving contraction of \(F\) if \(\int_{0}^{x} G(s) d s \leq \int_{0}^{x} F(s) d s\) for all \(x \in [0,1]\), where the inequality holds with equality at \(x=1\).} which we denote by $\mathcal{M}\left(F\right)$.

Given this, we specify that the sender's cost of acquiring information is posterior-mean measurable: the cost of acquiring any $G \in \mathcal{M}(F)$ is
\[C\left(G\right) = \kappa \int_{0}^{1}c\left(x\right)dG\left(x\right) \text{,}\]
where $c$ is a twice-differentiable, strictly convex function that is bounded on $\left(0,1\right)$ and satisfies $c\left(\mu\right) = 0$; and $\kappa > 0$ is a scaling parameter. Note that this cost function has the desirable feature that it is monotone with respect to informativeness: if $G$ is a MPC of $\hat{G}$, i.e., corresponds to a less informative Blackwell experiment than $\hat{G}$, $C\left(G\right) \leq C\left(\hat{G}\right)$. One specific example of such a cost function is that in which $c\left(x\right) = \left(x-\mu\right)^2$, which we use in Section \ref{uniformsec}'s example.\footnote{We impose that the cost of acquiring information is mean-measurable for simplicity of exposition. By identifying with the prior $F$ a joint distribution $\Phi\left(x_1, x_2, x_3, \dots, x_n\right)$ over the first $n$ non-centered posterior moments, we could instead that the cost of acquiring information is functional $C\left(\Upsilon\right) = \kappa \int c d\Upsilon$ where $\Upsilon$ is a fusion (the $n$-dimensional generalization of an MPC) of $\Phi$ and $c \colon \mathbb{R}^{n} \to \mathbb{R}_{+}$ is a convex function. As we note in the \href{https://whitmeyerhome.files.wordpress.com/2022/08/costly_evidence_and_discretionary_disclosure_supplement_v2.pdf}{Supplementary Appendix}, our main findings all continue to hold in this generalization.}

It is only the posterior mean that is strategically relevant for $R$, and therefore the same holds for $S$. Because of this, we limit the space of messages available to $S$ to be $M \coloneqq \left[0,1\right] \cup \left\{m_{\varnothing}\right\}$, where $m_{\varnothing}$ corresponds to non-disclosure (non-certification). This is a game of verifiable disclosure: if $S$ acquires posterior mean $x$, her set of available messages is $M_{x} \coloneqq \left\{x, m_{\varnothing}\right\}$. Disclosure of $x$ is costly, so we can think of the choice of message $x$ as paying a certifier cost $\gamma > 0$ to verify that the posterior mean is indeed $x$ and the choice of message $m_{\varnothing}$ as $S$ declining to have her claim certified. We specify that the sender's overall payoff is additively separable in $u_S$, the cost of certification ($\gamma$), and the cost of acquiring information.

We are interested in the effects of transparency in this environment and, therefore, look at two different cases:
\begin{enumerate}[noitemsep,topsep=0pt]
    \item \textbf{Covert Evidence Acquisition} $R$ does not observe $G$ and $G$ cannot be certified.
    \item \textbf{Overt Evidence Acquisition} $R$ observes $G$. 
\end{enumerate}
The latter case is ``more transparent'' than the former: $R$ observes $S$'s information gathering activities, no matter whether $S$ gets the outcome certified.

The timing of the game is as follows:
\begin{enumerate}[noitemsep,topsep=0pt]
    \item $S$ acquires information by choosing some distribution of posterior means $G \in \mathcal{M}\left(F\right)$.
    \item $S$ observes the realization $x$ from $G$ then chooses whether to get $x$ certified (send message $x$) and incur cost $\gamma$ or not (send message $m_{\varnothing}$).
    \item $R$ observes $S$'s message and $G$ if evidence acquisition is overt.
    \item $R$ takes action $a$ and payoffs realize.
\end{enumerate}

We study perfect Bayesian equilibria of the game. A formal definition of the equilibria, for both covert and overt evidence acquisition, can be found in Appendix \ref{eqdef}.

\section{Costless Evidence Benchmark}\label{bench}

We begin by discussing a natural benchmark; the case when $\kappa = 0$, and, therefore, it is costless for $S$ to acquire information. First, let us look at the overt acquisition case. Suppose first that the certification fee is not too high: $\gamma < 1- \mu$; in which case, given a conjectured (expected) value of non-disclosed evidence, $\alpha$, the sender's payoff from acquiring posterior $x$ is
\[V\left(x\right) = \begin{cases}
\alpha, \quad &\text{if} \quad 0 \leq x < \alpha + \gamma\\
x-\gamma, \quad &\text{if} \quad \alpha + \gamma \leq x \leq 1
\end{cases}\text{,}\]
where $\alpha \leq \mu$.

This function is convex and piecewise linear, with one inflection point at $\alpha + \gamma < 1$. Thus, in any best response by $S$, she will not pool any states below $\alpha + \gamma$ with those above when she acquires information. She does not want to contract the distribution across the inflection point as doing so would strictly lower her payoff. Consequently, in any equilibrium the disclosure must be a truncation: the sender will not disclose any evidence below the equilibrium inflection point $\alpha^{*} + \gamma$ of the value function. Strict log-concavity of $F$ (\cite{bag}) ensures that the state of indifference, and therefore $\alpha^{*}$, is unique:
\begin{lemma}\label{firstlemma}
With covert evidence acquisition and a low certification cost ($\gamma < 1-\mu$), in any equilibrium, the expected value of the state conditional on non-certification, $\alpha^{*}$, solves \[\label{star}\tag{$\star$}\alpha^{*} = \mathbb{E}_{F}\left[\left. \theta \right| \theta \leq \alpha^{*}+\gamma\right]\text{.}\]
\end{lemma}

On the other hand, when obtaining evidence is free, the sender's payoff is linear on $[\alpha^* + \gamma, 1]$, and hence she is indifferent as to what evidence she acquires about states above $\alpha^{*}+\gamma$. The receiver, of course, is not indifferent and prefers more information. As a result, there is a large multiplicity of equilibria, which can be Pareto ranked.
\begin{proposition}\label{firstprop}
With covert evidence acquisition and a low certification cost ($\gamma < 1-\mu$), any evidence acquisition protocol in which the sender does not pool states above $x^{*} \coloneqq \alpha^{*} + \gamma$ (where $\alpha^*$ is defined in \ref{star}) with those below, then certifies a posterior, $x$, if and only if $x \geq x^*$, is an equilibrium. There are no other equilibria. 

These equilibria can be Pareto ranked: the Pareto-best equilibrium is that in which the sender acquires full information and discloses above $x^{*}$. The Pareto-worst equilibrium is that in which the distribution of posterior means above $x^*$ is degenerate on $\mathbb{E}\left[\left.\theta\right|\theta \geq x^*\right]$.
\end{proposition}
It is clear that the Pareto-maximal payoff is uniquely induced (since it is an MPS of all other equilibrium distributions of posteriors) by the stated equilibrium.\footnote{By uniquely we mean unique up to the distribution of non-certified posteriors, which are behaviorally and payoff-irrelevant.} Moreover, the receiver's (and sender's) payoffs are precisely those that they would obtain if the information were exogenously $F$. \textit{Viz.,} the receiver obtains the same payoff as if the sender were fully informed about the state exogenously. Similarly, the Pareto-minimal payoff is also uniquely induced (since it is an MPC of all other equilibrium distributions). In this equilibrium, the sender (and receiver) learn only whether the state is above $x^*$, nothing more. 

When certification is expensive: $\gamma \geq 1-\mu$, in any equilibrium there is no certification.
\begin{proposition}\label{secondprop}
With covert evidence acquisition and a high certification cost ($\gamma \geq 1-\mu$), there is no certification in equilibrium.
\end{proposition}
Because there is no certification, these equilibria are all equivalent in the sense that there is no information transmission and the sender's and receiver's payoffs are unaffected by the distribution of posterior means acquired by the sender.

When obtaining evidence is overt, the analysis is equally simple: 
\begin{proposition}\label{lastbenchprop}
With overt evidence acquisition, there is an equilibrium in which the sender acquires no information and never certifies. In every equilibrium, the sender acquires no posteriors strictly greater than $\min\{\mu + \gamma,1\}$ and never certifies.
\end{proposition}
For all intents and purposes, all of the equilibria with overt evidence acquisition are the same and equivalent to the one in which the sender chooses the degenerate distribution, i.e., acquires zero evidence and never certifies. In each, there is no information transmission and the sender does not waste any resources on certification. The partial commitment power granted to the sender by transparency hurts the receiver. The sender may now credibly commit to effectively learn nothing about the state, which she does in every equilibrium. As we will see, this result persists when evidence is costly ($\kappa > 0$). 

\section{Main Analysis}\label{main}
Now we introduce a cost to acquiring information, i.e., the cost parameter is $\kappa > 0$.

\subsection{Covert Evidence Acquisition}

Our first result notes that at most one piece of certified evidence is followed by certification (with positive probability) in any equilibrium.
\begin{lemma}\label{mainlemma1}
In any equilibrium, the distribution of posterior means acquired by the sender, $G$, has support on at most one posterior that is certified with positive probability.
\end{lemma}
Key to this lemma is the fact that $S$'s payoff is linear (gross of the cost of obtaining the belief) in $R$'s posterior $x$, but the cost of obtaining a belief $x$ is strictly convex. Thus, by learning less collapsing beliefs to their average, $S$ can save on costs while not affecting the direct payoff from certification. 
\begin{lemma}\label{mainlemma2}
In any equilibrium, the distribution of posterior means acquired by the sender, $G$, has support on at most one posterior that is not certified with positive probability.
\end{lemma}
In Lemma \ref{mainlemma1}, it was the linearity of the gross payoff that rendered a single certified posterior optimal. Here, the gross payoff from a non-certified posterior is just a constant--recall that it is the conjectured posterior following non-certification (message $m_{\varnothing}$)--which again makes coarse(r) learning superior for the sender. Combining these lemmas, we have the following proposition, which pins down coarse learning as a necessary feature of any equilibrium:
\begin{proposition}\label{simplify}
In any equilibrium, the distribution of posterior means acquired by the sender, $G$, has support on at most two points. The sender never randomizes following information acquisition: a posterior in support of $G$ is either certified or not.
\end{proposition}
We see that the costliness of obtaining evidence leads to coarse information. The sender optimally acquires at most two pieces of evidence: good news, which she finds worthwhile to get certified; or bad news, which she does not. This hints at our convergence result of the next section. Information is valuable for the receiver, but this is not internalized by the sender, who instead saves on her privately-incurred costs of obtaining evidence by omitting fine details.

Next, we turn our attention to equilibrium existence, before moving on to the more interesting (and economically relevant) discussion of properties of equilibria.
\begin{theorem}\label{maintheorem}
An equilibrium always exists. It is not necessarily unique.
\end{theorem}
Here is a sketch of the proof: if the cost of certification is sufficiently high ($\gamma \geq 1-\mu$), the result is immediate since there exists an equilibrium in which the sender acquires no information and certifies with probability $0$. Now let $\gamma < 1-\mu$. For a fixed (conjectured) value of non-certification, $\alpha \in \left[0,\mu\right]$, the sender's payoff as a function of her acquired posterior $x$ is 
\[V_{\alpha}\left(x\right) = \begin{cases}
\alpha - \kappa c\left(x\right), \quad &\text{if} \quad \alpha + \gamma > x\\
x-\gamma- \kappa c\left(x\right), \quad &\text{if} \quad \alpha + \gamma \leq x\\
\end{cases} \text{.}\]
In her evidence acquisition problem, the sender chooses a distribution $G_{\alpha}$ that solves 
\[\label{maxprob}\tag{$\triangle$}\max_{G \in \mathcal{M}\left(F\right)}\int V_{\alpha} \, dG \text{.}\]
We finish by showing that there must always exist a fixed point $\alpha$ that corresponds to the sender's optimally acquired $G_{\alpha}$.

At equilibrium there are four potential types of information acquisition. First, the sender may acquire coarse evidence about whether the state is high or low (Figure \ref{figsub1}). There, the sender's evidence acquisition corresponds to a truncation of the distribution: she is sure whether the state is above some threshold, but nothing more. This occurs if the cost of acquiring information ($\kappa$) is sufficiently low and the certification cost ($\gamma$) is also low (but not $0$). In the parlance of \cite{martini} or \cite{bipool}, this is a monotone partitional evidence structure. Second, the sender may acquire evidence but occasionally conflate high states with low ones and \textit{vice versa} (Figure \ref{figsub2}). This may occur (depending on the cost function) if $\kappa$ and $\gamma$ are moderately low. In the parlance of \cite{bipool}, this is a ``bi-pooling'' evidence structure.

Third, the sender may acquire no evidence but always certify it (Figure \ref{figsub3}). This occurs if $\kappa$ is high but $\gamma$ is low. Non-certification is off-path and is met with extreme pessimism ($0$ belief) if the receiver observes it. Fourth, the sender may acquire no evidence and never certify it (Figure \ref{figsub4}). This occurs if $\kappa$ and $\gamma$ are high. In both of these cases, the sender collapses the entire distribution to its barycenter (mean), learning nothing.

\begin{figure}
\centering
\begin{subfigure}{.5\textwidth}
  \centering
  \includegraphics[scale=.15]{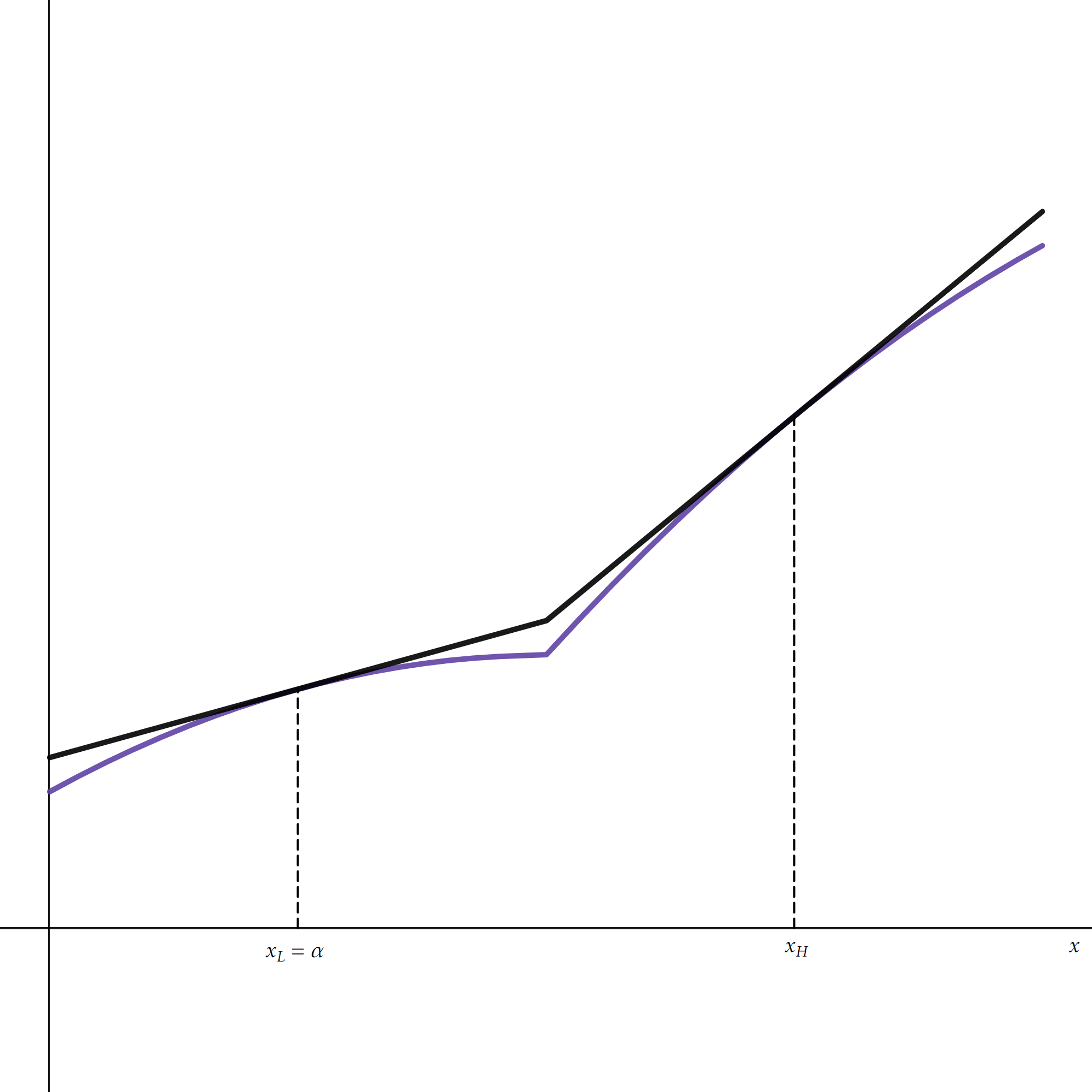}
  \caption{Low $\kappa$ and $\gamma$.}
  \label{figsub1}
\end{subfigure}%
\begin{subfigure}{.5\textwidth}
  \centering
  \includegraphics[scale=.15]{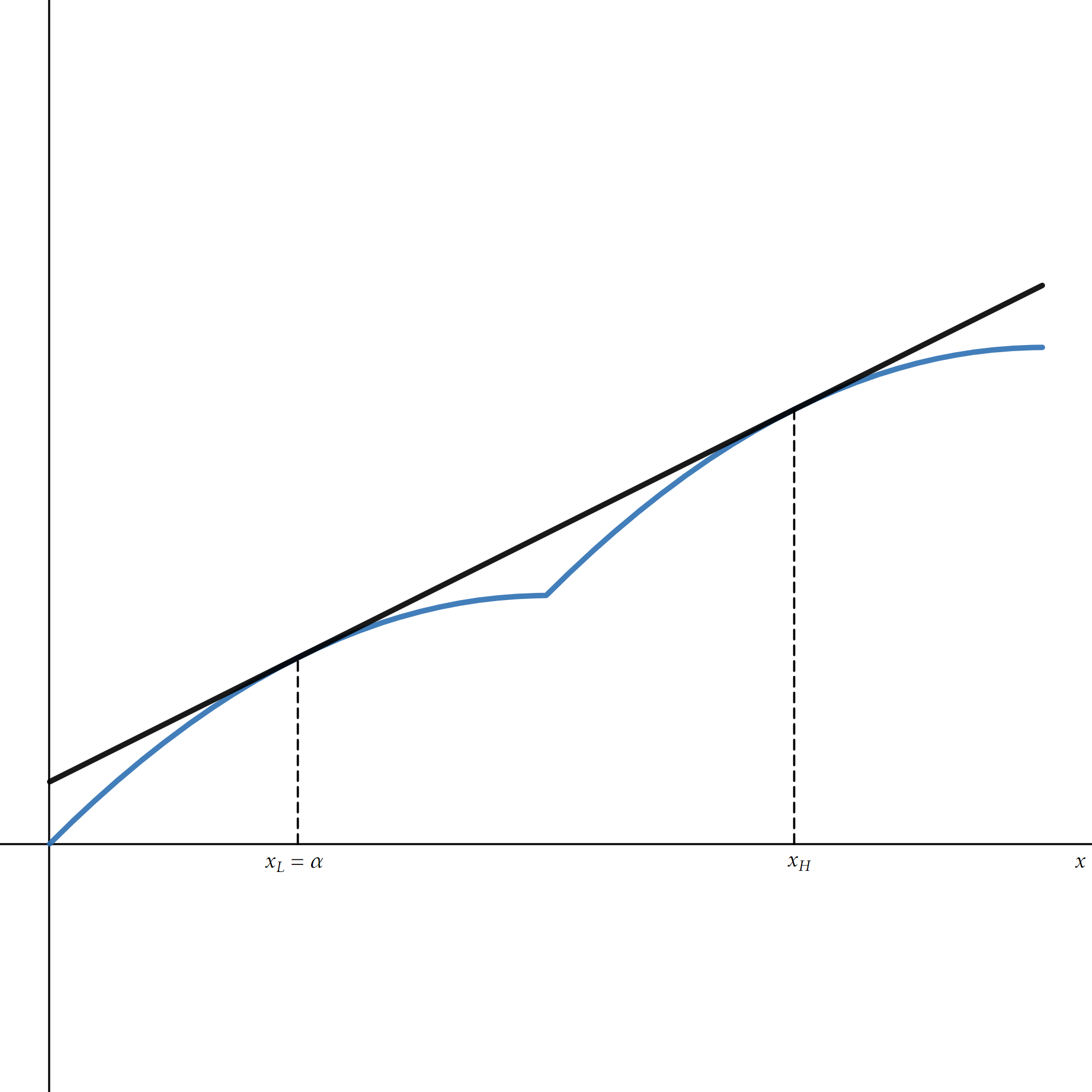}
  \caption{Moderate $\kappa$ and low $\gamma$.}
  \label{figsub2}
\end{subfigure}
\par
\bigskip
\par
\bigskip
\par
\begin{subfigure}{.5\textwidth}
  \centering
  \includegraphics[scale=.15]{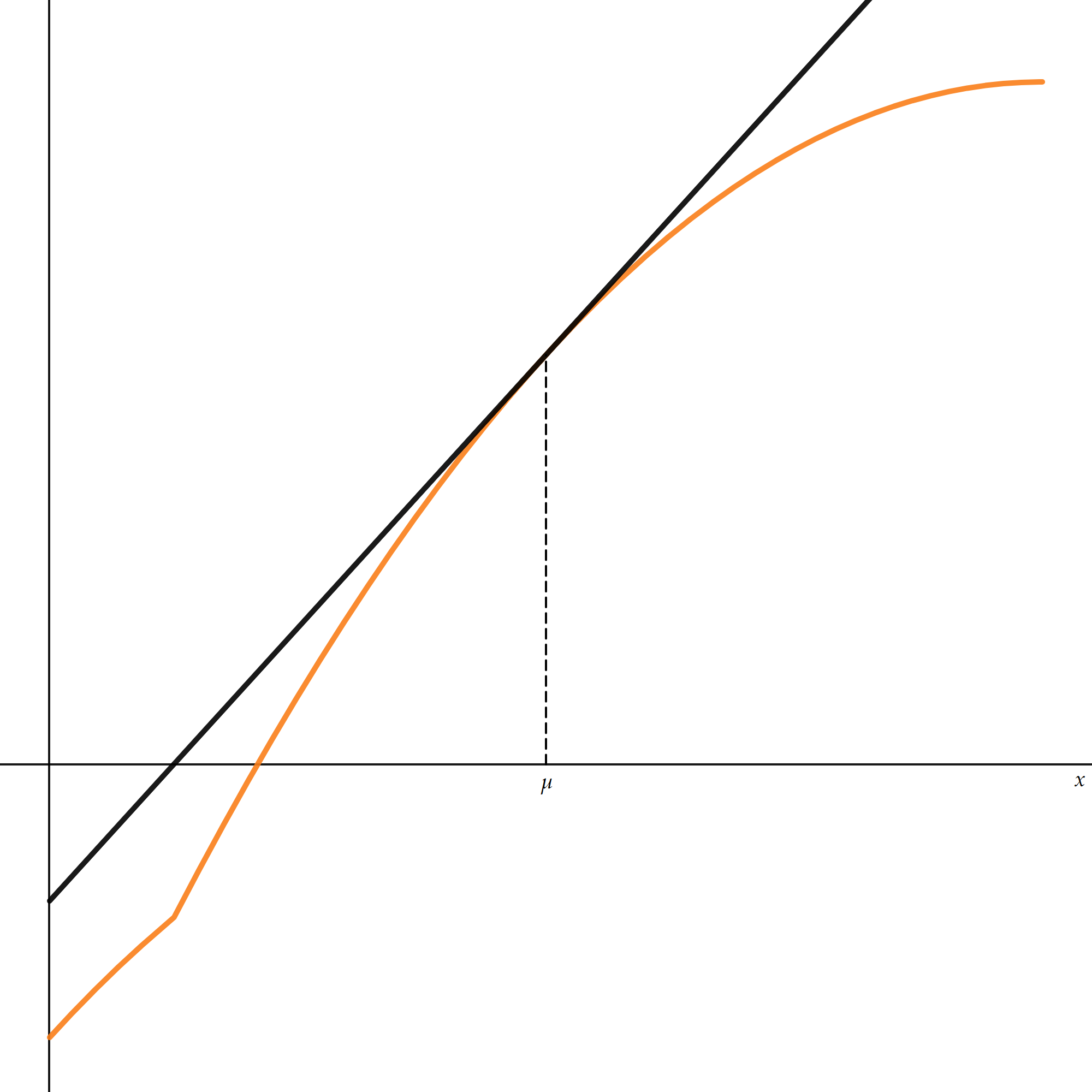}
  \caption{High $\kappa$ and low $\gamma$.}
  \label{figsub3}
\end{subfigure}%
\begin{subfigure}{.5\textwidth}
  \centering
  \includegraphics[scale=.15]{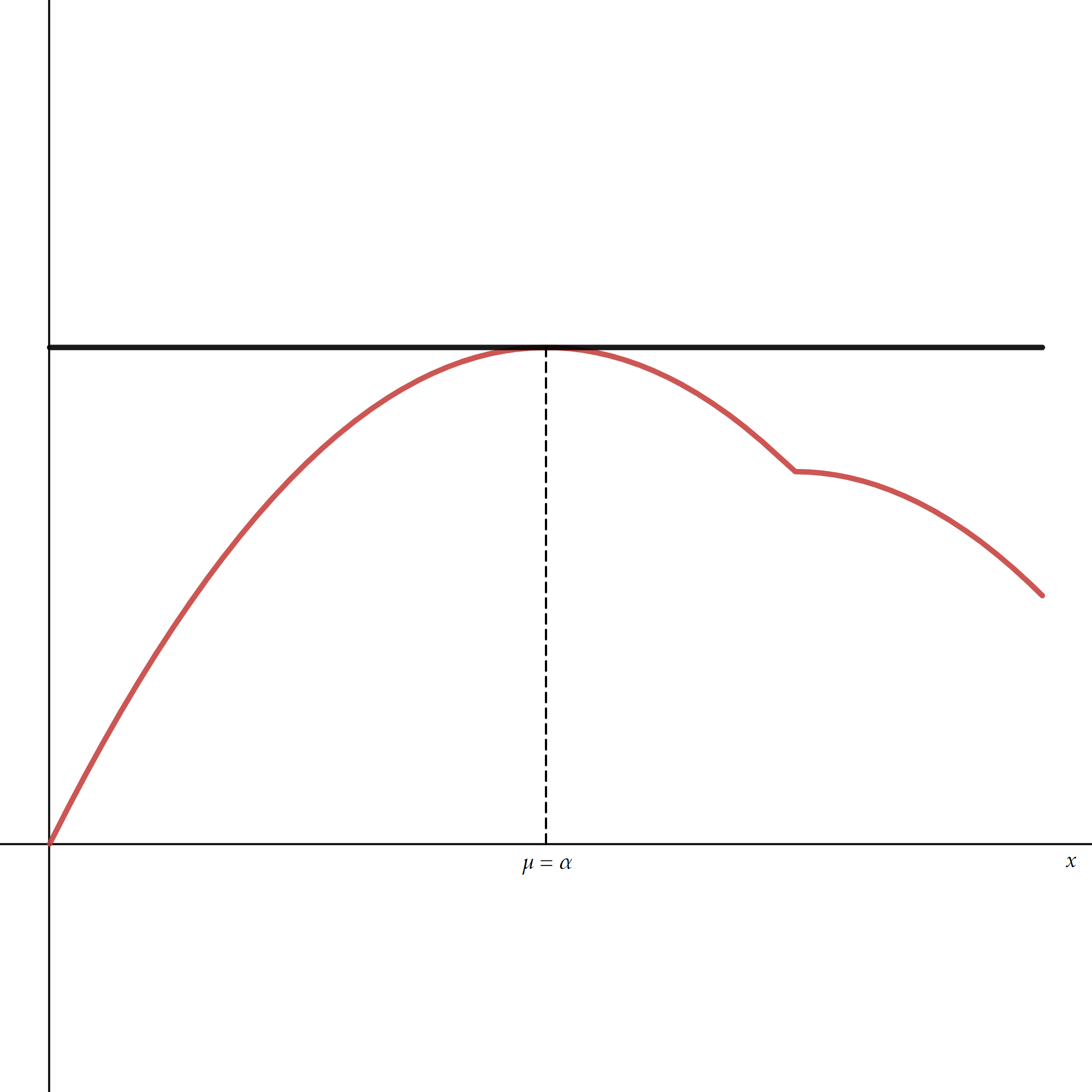}
  \caption{High $\kappa$ and $\gamma$.}
  \label{figsub4}
\end{subfigure}
\caption{\textbf{The Four Varieties of Covert Equilibria:} The equilibrium payoff functions as functions of the acquired posterior mean ($x$) given the conjectured non-certification value ($\alpha$) are the colored curves. The dual in this optimization problem (the ``price function'' of \cite{martini}), in black, determines the sender's optimal evidence.}
\label{fig2}
\end{figure}

\subsection{Some Properties of (Covert) Equilibria}

No matter the cost functional, the prior, or the cost of certification, if information is sufficiently costly, there exists an equilibrium in which the sender does not acquire any information, and there exist no equilibria in which the sender acquires information.

\begin{proposition}\label{costofinfoprop}
For any prior $F$, cost functional $C$, and certification cost $\gamma > 0$, if $\kappa$ is sufficiently high, there exist no equilibria with evidence acquisition. 
\end{proposition}
That is, for any fixed combination of the other primitives, there exists a threshold cost of information $\bar{\kappa}$ such that if $\kappa$ exceeds $\bar{\kappa}$ the only equilibria are those in which the sender does not obtain any evidence.

Moreover, as we illustrate in the proof, for $\kappa$ sufficiently large, there is a multiplicity: in any equilibrium the sender acquires no evidence, but there is an equilibrium in which the sender never gets the evidence (or lack thereof) certified \textit{and} an equilibrium in which the sender gets the evidence certified.

In contrast, unless $\gamma \geq 1-\mu$, if $\kappa$ is sufficiently small, the sender acquires information then discloses if the evidence is good.
\begin{proposition}\label{costofinfoprop2}
For any prior $F$, cost functional $C$, and certification cost $0 < \gamma < 1-\mu$, there exists a threshold cost of information $\ubar{\kappa} > 0$ such that if $\kappa \leq \ubar{\kappa}$, all equilibria must involve information acquisition.
\end{proposition}
The effect of the certification on the existence of equilibria in which the sender acquires information is more ambiguous. Roughly, as $\gamma$ decreases, it makes it harder to sustain equilibria in which the sender neither acquires evidence nor gets it certified. However, it makes it easier to sustain equilibria in which the sender does not acquire evidence but nevertheless gets it certified. In fact, when $\gamma = 0$, this is the only equilibrium that exists:
\begin{proposition}\label{nodisccostprop}
If certification is costless, $\gamma = 0$, but evidence gathering is costly, $\kappa >0$, the unique equilibrium is that in which the sender acquires no evidence but gets it certified.
\end{proposition}
This proposition reveals a key difference between this paper's setting--with costly (flexible) evidence acquisition--and the exogenous information environment. In the latter, the receiver's welfare is strictly decreasing as the disclosure (verification) cost increases since the equilibrium level of information transmitted is strictly decreasing. Here, we see that costless certification is the worst case for the receiver. Due to the endogeneity of the sender's evidence, it is only when the certification cost is strictly positive that the sender faces sufficient incentives to acquire information.

As we explain in the introduction, the force driving the result is the local convexity in the payoff from acquiring evidence at the conjectured ``no-certification belief'' $\alpha$. Figure \ref{fig3} depicts this for various cost functions.

\begin{figure}
    \centering
    \includegraphics[scale=.15]{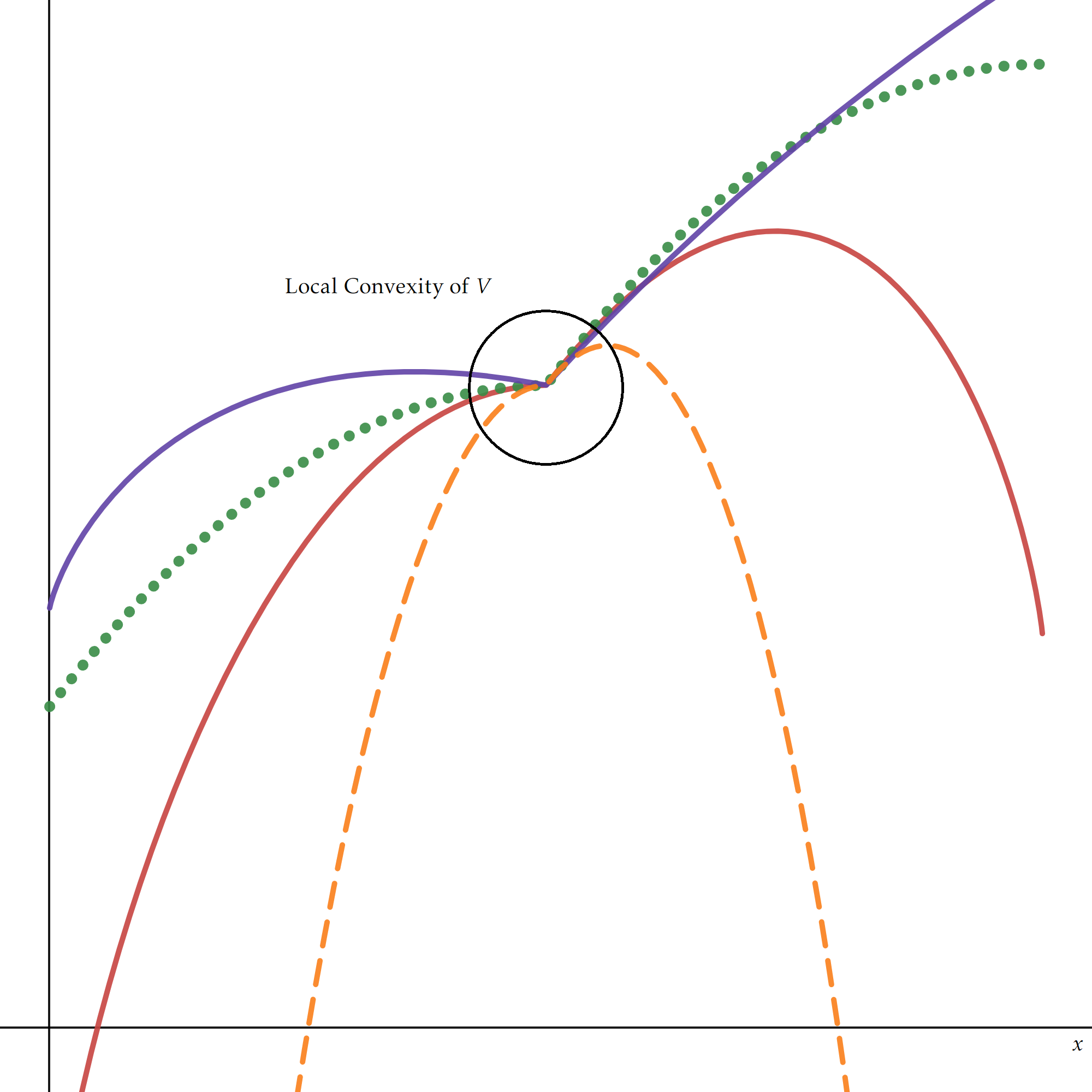}
    \caption{The Local Convexity of $V\left(x\right)$ at $\alpha$.}
    \label{fig3}
\end{figure}

Naturally, the certification cost cannot be too high: if $\gamma \geq 1-\mu$, there is no evidence acquisition (and hence no information transmission) at equilibrium. However, provided $\gamma$ is in an intermediate range, and the cost of acquiring information is not too high (refer to Proposition \ref{costofinfoprop2}), there is information transmission in any equilibrium, a strict improvement for the receiver. Summing things up,
\begin{corollary}
The receiver prefers a strictly positive certification cost to no certification cost. If the cost of acquiring evidence is sufficiently low, this preference is strict.
\end{corollary}

\subsection{Overt Evidence Acquisition}
The presence of evidence acquisition costs does not change the effects of transparency. Now, however, the equilibrium is unique. The sender acquires no information, obtaining her maximal payoff, to the receiver's chagrin:
\begin{proposition}\label{overtbad}
The unique equilibrium is for the sender to acquire no information and not disclose.
\end{proposition}
\begin{proof}
The maximal payoff the sender can obtain in any equilibrium is $\mu$, the expected state. Thus, in any equilibrium she must get this payoff, because otherwise she could deviate by acquiring no information and not certifying. This is credible, since her evidence acquisition (at least her distribution of posteriors) is public. Because evidence acquisition is costly, a payoff of $\mu$ corresponds uniquely to acquiring no evidence followed by $m_{\varnothing}$.
\end{proof}

Of course, by Jensen's inequality, this result continues to hold as long as long as the sender's reduced-form payoff as a function of the receiver's posterior is strictly concave. An easy sufficient condition for this is that the sender's reduced-form payoff (gross of the cost of obtaining evidence) is weakly concave. Another sufficient condition is that the second derivative of the sender's reduced-form payoff (gross of cost) is bounded on $\left[0,1\right]$ and the cost parameter $\kappa$ is sufficiently large.

In light of Proposition \ref{lastbenchprop}, which states that there is no information transmission when information is free and its acquisition is overt, this result is not surprising. A cost of acquiring evidence only discourages evidence acquisition, which gives the sender more of an incentive to forgo learning.

\subsection{Uniform Prior, Quadratic Cost Example}\label{uniformsec}

\begin{figure}
    \centering
    \includegraphics[scale=.15]{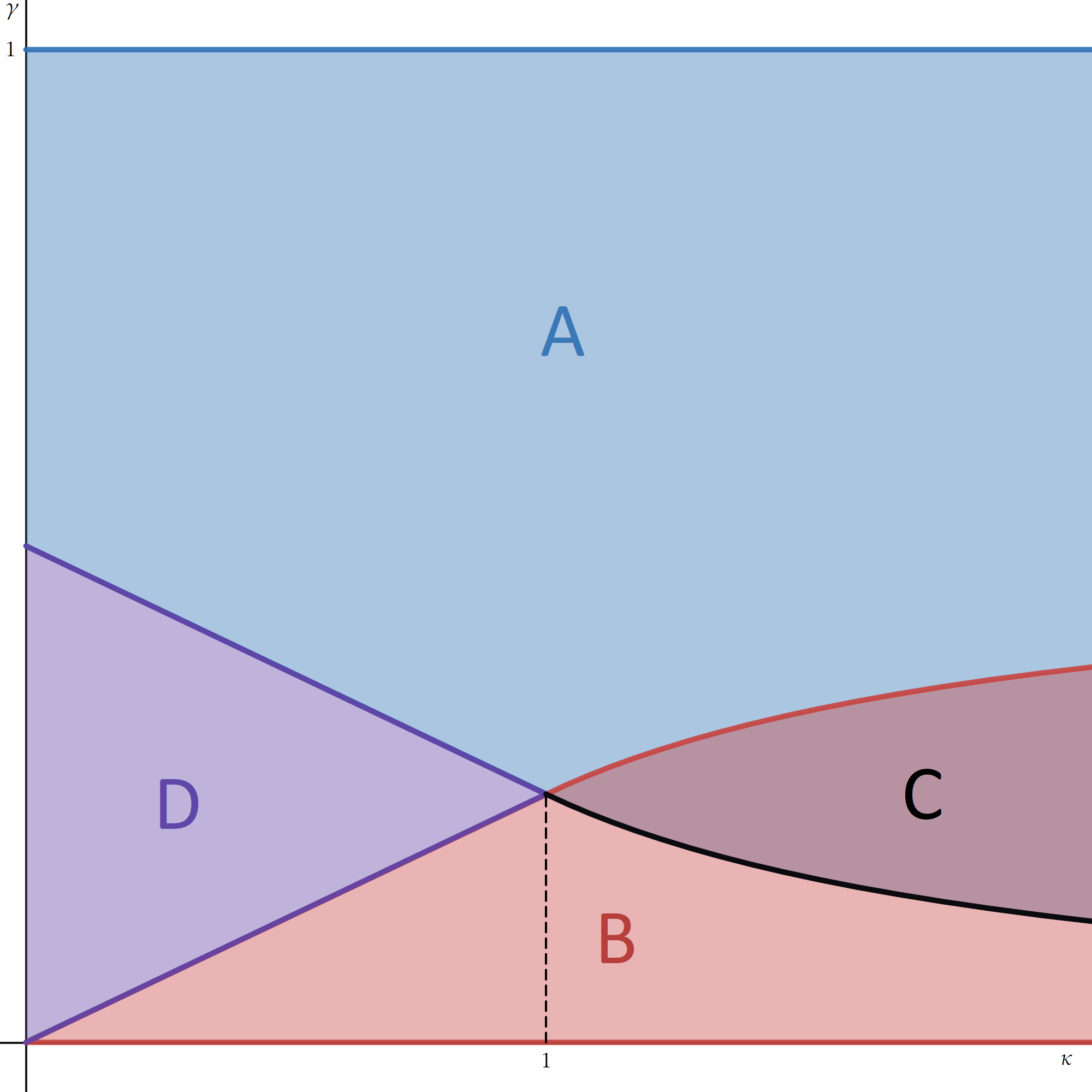}
    \caption{Division of the Parameter Space for the Uniform Prior, Quadratic Cost Example. In Regions $A$ and $B$, the sender acquires no evidence and either always does not disclose, or always discloses, respectively. In Region $C$, these equilibria coexist. In Region $D$, the sender acquires evidence at equilibrium.}
    \label{figregions}
\end{figure}
Next, we will briefly discuss an example in which the state's distribution is the standard uniform and the information cost is $c\left(x\right) = \left(x-\mu\right)^2$.
\begin{proposition}\label{upqc}
Except for a knife-edge case, there exists an equilibrium in which the sender acquires evidence if and only if $\kappa$ is sufficiently small ($\kappa \in \left[0,1\right]$) and $\gamma$ is neither too large nor too small ($\frac{\kappa}{4} \leq \gamma \leq \frac{1}{2}-\frac{\kappa}{4}$). Otherwise, the sender does not acquire evidence.
\end{proposition}
In Appendix \ref{app11}, we specify the exact parametric regions, which are depicted in Figure \ref{figregions}. When $\kappa \geq 1$, the two types of no-evidence-acquisition equilibria coexist for intermediate values of the certification cost ($\gamma \in \left[\frac{1}{4\kappa}, \frac{1}{2}-\frac{1}{4\kappa}\right]$). 

The key qualitative features of the equilibria in this example relate directly to our earlier propositions. The sender acquires evidence and there is information transmission at equilibrium, if and only if information is relatively cheap ($\kappa \leq 1$) and certification costs are relatively small, but not too small.

\section{Selecting the Pareto-Worst Equilibrium}\label{pareto}
In this section we discuss what happens to covert equilibria when the cost of learning vanishes, i.e., $\kappa \downarrow 0$. When $\gamma \ge 1-\mu$, there is no certification and hence no information transmission in any equilibrium; consequently, whether $\kappa$ is nonzero does not matter. We henceforth focus on the more interesting case in which $\gamma < 1-\mu$.

A distribution of posterior means $G \in \mathcal{M}(F)$ is \textbf{monotone partitional} if there exists $t \in [0,1]$ such that $\supp\left(F\right) \subseteq \left\{x_L(t), x_H(t)\right\}$ with $x_L(t) = \mathbb{E}\left[x \, | \, x \in [0,t]\right]$ and $x_H(t) = \mathbb{E}\left[x \, | \, x \in [t,1]\right]$. Say that an equilibrium is \textbf{monotone partitional} if in this equilibrium, the distribution of posterior means acquired by $S$ is monotone partitional. By Proposition \ref{costofinfoprop2}, when $\kappa$ is sufficiently small, the sender must acquire information. The following result characterizes the set of equilibria in this case. 

\begin{lemma}\label{lemma61}
For any $\gamma \in (0,1-\mu)$, there is a threshold $\Tilde{\kappa}\left(\gamma\right) > 0$ such that for all $\kappa \leq \Tilde{\kappa}\left(\gamma\right)$ any equilibrium is monotone partitional.
\end{lemma}
Intuitively, as $\kappa$ decreases, the sender's value function is ``less distorted'' by the cost. In particular, when $\kappa$ gets very small, it gets ``close'' to the piecewise linear and convex value function in the $\kappa=0$ benchmark. Any feasible distribution of posterior means that is not monotone partitional must involve contracting some high states and low ones to an intermediate state, and hence cannot be optimal for any sufficiently small $\kappa$.

If the equilibrium is monotone partitional and involves information acquisition, the sender solves
\[\max_{t}\left\{F\left(t\right)\left(\alpha - \kappa c\left(x_L\left(t\right)\right)\right) + \int_{t}^{1}adF\left(a\right) -  \left(1-F\left(t\right)\right)\left(\gamma + \kappa c\left(x_H\left(t\right)\right)\right) \right\} \text{,}\]
where $x_{L}\left(t\right) \coloneqq \mathbb{E}\left[\left.\theta\right|\theta \leq t\right]$ and $x_{H}\left(t\right) \coloneqq \mathbb{E}\left[\left.\theta\right|\theta > t\right]$.
Taking the FOC then setting $\kappa = 0$ and replacing $\alpha$ in the FOC with $x_L$, we obtain $x_L = t - \gamma $;
and so the limiting truncation point, $t$, is precisely that described in Equation \ref{star} in the costless evidence benchmark. Moreover, for arbitrary small $\kappa > 0$, the sender's posterior distribution is binary. Thus, the limit distribution is the Pareto-worst distribution described in Proposition \ref{firstprop}.
\begin{theorem}\label{thm62}
Take any sequence of equilibria as $\kappa \downarrow 0$. The limit equilibrium is the Pareto-Worst equilibrium when evidence is free.
\end{theorem}
Theorem \ref{thm62} reveals that there is a significant difference between ``free'' evidence and ``cheap evidence'' when information acquisition is covert. Even a minuscule (but strictly positive) cost to acquiring evidence leads the sender to acquire extremely coarse information. The finer details about the state are completely unimportant to the sender, who does not learn them no matter how cheap it is to do so. On the other hand, the receiver, who wants to ``match the state'' is hurt by any mismatch. This benefit is not internalized by the sender, so even though when information is vanishingly cheap it would be efficient to learn the state precisely, the sender does not do so.

In contrast, there is no real difference between the two concepts when evidence acquisition is overt: both when information is free and when it is cheap, the sender does not acquire information and does not disclose either. 

\section{Discussion}\label{discus}

The three main contributions of this paper are that 
\begin{enumerate}[noitemsep,topsep=0pt]
    \item When evidence acquisition is costly, transparency in the acquisition process is negative as it grants the sender a form of commitment to acquire less information.
    \item The receiver always weakly (and sometimes strictly) prefers a positive certification cost, in stark contrast to models with exogenous evidence.
    \item Cheap information is different from free information: allowing evidence acquisition costs to vanish selects for the Pareto-worst free-information equilibrium.
\end{enumerate}
It is natural to wonder whether these effects occur in the other ``warhorse'' disclosure model--namely, that of \cite{dye} and \cite{jungkwon}. In their paradigm, the disclosure cost of \cite{verr1} is replaced with exogenous failure: with some probability $\gamma \in \left(0,1\right)$ the sender obtains no evidence. They do: we argue in the \href{https://whitmeyerhome.files.wordpress.com/2022/08/costly_evidence_and_discretionary_disclosure_supplement_v2.pdf}{Supplementary Appendix} that these findings persist in \citeauthor{dye}'s setting. The commitment engendered by transparency, the optimality of coarse information when evidence is costly, and detrimental effects of unraveling on evidence acquisition continue to hold, and, therefore, continue to drive the results.

Beyond this, other interesting questions and extensions remain. For one, we assume that the sender cannot falsify or hide evidence. In many applications--our leading example of an asset manager is one--this is reasonable: deception, omission, and falsification are prohibited by law. However, there are other applications in which hiding evidence may be realistic. It is possible that this may generate instances in which overt evidence acquisition is superior to covert. It might also be interesting to look at behavior from the certifier's point of view. How would the certifier choose her fee--and possibly other instruments, coarse classifications, say--to maximize profits?

\bibliography{sample.bib}

\appendix

\section{Equilibrium Definition} \label{eqdef}
Let $\Pi$ denote the set of signals, i.e., the set of measurable maps $\pi\colon \Theta \to \Delta\left(X\right)$. For every $\pi \in \Pi$, let $G_{\pi}$ denote the distribution of posterior means induced by $\pi$. $S$'s strategy set, $\Pi \times D$, is the set of pairs $(\pi, d)$, where $\pi \in \Pi$ and $d\colon X \times \Pi \to \Delta \left(M\right)$ is such that $d\left(\left.m \right| x\right) > 0$ implies $m \in M_x$. $R$'s strategy set is the set of functions $\alpha\colon M \times \Pi \to A$. Finally, because only the posterior mean of the state is strategically relevant to $R$, $R$'s beliefs about the state can be summarized by a function $p\colon M \times \Pi \to [0,1]$ whose codomain is the set of posterior means of the state. 

We say that $(\pi, d, \alpha, p)$ is a \textbf{covert equilibrium} if
\begin{enumerate}[label={(\arabic*)},noitemsep,topsep=0pt]
    \item for all $\hat{\pi} \in \Pi$, $p\left(x, \hat{\pi}\right) = \mathbb{E}\left[\theta \left| x \right.\right]$ for all $x \in [0,1]$, and $p\left(m_{\varnothing}, \hat{\pi}\right) = \mathbb{E}\left[\theta \left| m_{\varnothing}, \pi \right.\right]$;
    \item for all $m \in M$ and $\hat{\pi} \in \Pi$, $\alpha\left(m, \hat{\pi}\right) = p\left(m, \hat{\pi}\right)$;
    \item $\left(\pi,d\right) \in \argmax_{\left(\hat{\pi}, \hat{d}\right)} \mathbb{E}\left[\hat{d}\left(x | x \right) \left(\alpha\left(x, \hat{\pi}\right) - \gamma\right) + \hat{d}\left(m_{\varnothing} | x \right) \left(\alpha\left(m_{\varnothing}, \hat{\pi}\right)\right)\right] - C(G_{\hat{\pi}})$.
\end{enumerate}
$(\pi, d, \alpha, p)$ is an \textbf{overt equilibrium} if (2) and (3) hold, and (1) is replaced by 
\begin{enumerate}[label={(\arabic*')},noitemsep,topsep=0pt]
    \item for all $\hat{\pi} \in \Pi$, $p(x, \hat{\pi}) = \mathbb{E}\left[\theta \left| x \right.\right]$ for all $x \in [0,1]$, and $p(m_{\varnothing}, \hat{\pi}) = \mathbb{E}\left[\theta \left| m_{\varnothing}, \hat{\pi} \right.\right]$.
\end{enumerate}

\section{Omitted Proofs}\label{omitproof}

\subsection{Lemma \ref{firstlemma} Proof}
\begin{proof}
The piecewise linearity of $V$ implies that any equilibrium must be above the threshold form: given any conjectured expected value of the state $\alpha^{*}$ the sender will not pool states above $\alpha^{*} + \gamma$ with states below in the evidence acquisition stage. Given this, and the linearity of the payoff function above and below $\alpha^{*} + \gamma$, any equilibrium is characterized by a posterior threshold $x^{*}$ that solves 
\[\xi\left(x^{*}\right) \coloneqq x^{*} - \frac{\int_{0}^{x^{*}}xdF\left(x\right)}{F\left(x^*\right)} = \gamma\text{.}\]
Setting $x^* = \alpha^*+\gamma$ produces Equation \ref{star}. By Theorem 5 in \cite{bag}, strict log-concavity of $F$ means that $\xi\left(x^{*}\right)$ is strictly increasing and therefore $x^{*}$ is unique. 
\end{proof}

\subsection{Proposition \ref{secondprop} Proof}
\begin{proof}
Again, the log-concavity of $F$ implies that $\xi\left(x^{*}\right) - \gamma \leq 0$ for all $x^{*} \in \left[0,1\right]$. Given this, there is no equilibrium in which the receiver's belief upon seeing no certification $\alpha \neq \mu$. Evidently, if $\alpha = \mu$, the sender prefers to not disclose any posterior.
\end{proof}
\subsection{Proposition \ref{lastbenchprop} Proof}
\begin{proof}
The sender's maximal payoff is $\mu$. Consequently, in any equilibrium she must be able to obtain this payoff, since she could always deviate by overtly acquiring the degenerate distribution on the prior $\delta_{\mu}$ then not disclosing. On the other hand, she cannot overtly acquire any posteriors with support strictly larger than $\mu + \gamma$ since otherwise she would want to deviate at an interim point and disclose those high beliefs.
\end{proof}
\subsection{Lemma \ref{mainlemma1} Proof}
\begin{proof}
Suppose for the sake of contradiction that there are at least two posteriors $x_1$ and $x_2$ in support of $G$ that are disclosed with positive probability. At $x_1$ and $x_2$ it must, therefore, be weakly optimal for $S$ to get each certified. However, due to the strict convexity of $c\left(\cdot\right)$, $S$ can deviate by acquiring (then disclosing) their average instead. Here we are implicitly assuming that $x_1$ and $x_2$ each realize with positive probability under $G$, but it is easy to see that the analogous deviation works when $G$ is atomless on of the set of posteriors that disclosed with positive probability.
\end{proof}

\subsection{Lemma \ref{mainlemma2} Proof}
\begin{proof}
From Lemma \ref{mainlemma1}, there is at most one posterior in support of $G$ that is certified with positive probability. Moreover, it is obvious that this belief must be certified with probability $1$. If there are multiple posteriors that are not certified, the receiver clearly prefers to collapse these posteriors to their barycenter, thereby saving on costs.
\end{proof}

\subsection{Theorem \ref{maintheorem} Proof}
We first establish two auxiliary results.

\begin{lemma}\label{uniqids}
    For every $\alpha \in [0, \mu]$, the solution to problem \ref{maxprob} is unique. 
\end{lemma}

\begin{proof}
    For each $\alpha \in [0, \mu]$, the proof of Lemma \ref{mainlemma1} and \ref{mainlemma2} together indicate that any solution to problem \ref{maxprob}, denoted by $G_{\alpha}$, must satisfy $\left|\supp{\left(G_{\alpha}\right)}\right| \le 2$. Fix some $\alpha \in [0, \mu]$. The statement is immediate when a solution to problem \ref{maxprob} is the degenerate distribution at $\mu$. Suppose $\left|\supp{\left(G_{\alpha}\right)}\right| = 2$. Because $V_\alpha(\cdot)$ is continuous for all $\alpha$, a ``price function'' $p$ \`{a} la \cite{martini} exists. There are two cases: one is that $p$ is affine on $[0,1]$ and tangent to $V_\alpha$ at two points, in which case we call $p$ a \emph{bi-tangent} in the parlance of \cite{bipool}; the other is that it is piecewise affine, and the two pieces have different slopes, in which case we say that $p$ is \emph{kinked}. 

    Suppose both $G_{\alpha}^{1}$ and $G_{\alpha}^{2}$, where $\supp{\left(G_{\alpha}^{1}\right)} \ne \supp{\left(G_{\alpha}^{2}\right)}$, solve the problem. Denote the corresponding price functions by $p_1$ and $p_2$, respectively. It cannot be that both $p_1$ and $p_2$ are bi-tangents, since a bi-tangent is unique whenever it exists. Now suppose $p_1$ is a bi-tangent, and $p_2$ is kinked. By construction, for each $i = 1,2$, $\supp{\left(F^i_\alpha\right)} = \left\{x : V_\alpha(x) = p(x) \right\}$. Then because $p_2$ is convex, it must be that \[\min\left\{\supp{\left(G_{\alpha}^{1}\right)}\right\} < \min\left\{\supp{\left(G_{\alpha}^{2}\right)}\right\} < \max\left\{\supp{\left(G_{\alpha}^{2}\right)}\right\} < \max\left\{\supp{\left(G_{\alpha}^{1}\right)}\right\}.\] But then by Lemma 1 in \cite{bipool}, there does not exist a MPC of $F$ whose support is $\supp{\left(G_{\alpha}^{1}\right)}$; a contradiction. 

    Finally, suppose both $p_1$ and $p_2$ are kinked. WLOG, assume that $\min\left\{\supp{\left(G_{\alpha}^{1}\right)}\right\} < \min\left\{\supp{\left(G_{\alpha}^{2}\right)}\right\}$. Consequently, the ``kink'' in $p_1$ must be strictly higher than the ``kink'' in $p_2$, and hence $\max\left\{\supp{\left(G_{\alpha}^{1}\right)}\right\} < \max\left\{\supp{\left(G_{\alpha}^{2}\right)}\right\}$. Then because $V_\alpha$ is piecewise strictly concave, $p'_1(x) < p'_2(x)$ for all $x \in [0,1]$,\footnote{If $p_i$, $i=1,2$, is not differentiable at $x$, we define $p'_i(x)$ to be its right-derivative.} and consequently $p_1(x) > p_2 (x)$ for $x \ge \min\left\{\supp{\left(G_{\alpha}^{2}\right)}\right\}$. But this means that
    \[p_2\left(\max\left\{\supp{\left(G_{\alpha}^{1}\right)}\right\}\right) < p_1\left(\max\left\{\supp{\left(G_{\alpha}^{1}\right)}\right\}\right) = V_\alpha\left(\max\left\{\supp{\left(G_{\alpha}^{1}\right)}\right\}\right),\]
    a contradiction. 
\end{proof}


To state the next result, let $L(\alpha) \coloneqq \min\left\{\supp{\left(G_{\alpha}\right)}\right\}$.

\begin{lemma} \label{ctnslow}
    Suppose $\left|\supp{\left(G_{\alpha}\right)}\right| =2$ for all $\alpha \in [0, \mu]$. Then the map $\alpha \mapsto L(\alpha)$ is continuous.    
\end{lemma}

\begin{proof}
    The objective function of problem \ref{maxprob} is continuous in both $\alpha$ and the choice variable $G$. The constraint correspondence is constant in $\alpha$ and hence continuous, and $\mathcal{M}$ is compact (endowed with the $L^1$ norm). By Berge's Maximum Theorem\footnote{For example, Theorem 17.31 in \cite{aliprantis}.}, the solution correspondence is upper hemicontinuous in $\alpha$; then by Lemma \ref{uniqids}, the unique solution $G_{\alpha}$ is continuous in $\alpha$. Then because $\left|\supp{\left(G_{\alpha}\right)}\right| =2$ for all $\alpha \in [0, \mu]$, $L(\alpha)$ must also be continuous. 
\end{proof}

By the duality approach of \cite{martini}, it is easy to see that (i) if $\mu \geq \gamma$ and $\left|\supp{\left(G_{0}\right)}\right| = \left|\supp{\left(G_{\mu}\right)}\right| =2$, $\left|\supp{\left(G_{\alpha}\right)}\right| =2$ for all $\alpha \in [0, \mu]$; and (ii) if $\mu < \gamma$ and $\left|\supp{\left(G_{\mu}\right)}\right| =2$, $\left|\supp{\left(G_{\alpha}\right)}\right| =2$ for all $\alpha \in [0, \mu]$. Next we prove the following proposition, which immediately implies the theorem.
\begin{proposition}
There exists at least one fixed point $\alpha \in \left[0,\mu\right]$ such that there exists a solution to the maximization problem \ref{maxprob}, $G_{\alpha}$, that is either binary, with a lower point of support $\alpha$ or degenerate on $\mu$. This fixed point is not necessarily unique.
\end{proposition}
\begin{proof}
If $\gamma \geq 1-\mu$ the result is trivial, so set $\gamma < 1-\mu$. 
First let $\mu \geq \gamma$. Set $\alpha = 0$. If the optimizer $G_{0}$ is degenerate on $\mu$, we have thus found an equilibrium, and we are done. Similarly, if the optimizer $G_{\mu}$ for $\alpha = \mu$ is degenerate on $\mu$, this is also an equilibrium. Suppose now that $\left|\supp{\left(G_{0}\right)}\right| = \left|\supp{\left(G_{\mu}\right)}\right| =2$, we must have $\left|\supp{\left(G_{\alpha}\right)}\right| =2$ for all $\alpha \in [0, \mu]$. 
By definition, $L(\alpha)$, the lower point of support of $G_{\alpha}$, is weakly greater than $0$ and weakly smaller than $\mu$. Consequently, $L$ is a continuous self-map on $[0,\mu]$. Then by Brouwer's fixed point theorem,\footnote{For example, Corollary 17.56 in \cite{aliprantis}.} there exists $\alpha \in [0, \mu]$ such that $L(\alpha) = \alpha$. 

Now let $\mu < \gamma$. If the optimizer $G_{\mu}$ for $\alpha = \mu$ is degenerate on $\mu$, this is an equilibrium, and we are done. Suppose instead that $\left|\supp{\left(G_{\mu}\right)}\right| =2$; consequently, $\left|\supp{\left(G_{\alpha}\right)}\right| =2$ for all $\alpha \in [0, \mu]$. Now we can repeat the above argument and a fixed point necessarily exists.
\end{proof}

\subsection{Proposition \ref{costofinfoprop} Proof}
\begin{proof}
If $\gamma \geq 1- \mu$, the threshold is $\bar{\kappa} = 0$. Now let $0 < \gamma < 1-\mu$ and let the conjectured non-certified value be $\mu$. The line tangent to the curve $\mu - \kappa c\left(x\right)$ at $x = \mu$ is
\[\kappa c'\left(\mu\right) \left(\mu-x\right) + \mu \text{.}\]
This line intersects the curve $x - \gamma - \kappa c\left(x\right)$ on the interval $\mu + \gamma \leq x \leq 1$ if and only if the function \[\tau\left(x\right) \coloneqq \kappa c'\left(\mu\right) \left(\mu-x\right) + \mu - x + \gamma + \kappa c\left(x\right)\text{,}\]
has a root in $\left[\mu + \gamma, 1\right]$. Next, \[\tau'\left(x\right) = \kappa \left(c'\left(x\right) - c'\left(\mu\right)\right) - 1 \text{.}\]
The term in parenthesis is strictly positive since $c$ is strictly convex and $\gamma > 0$. Accordingly, for all $\kappa$ sufficiently large, $\tau'\left(x\right)$ is strictly positive. Moreover, $\tau\left(\mu+\gamma\right) > 0$ and so $\tau$ has no root in the required interval, which allows us to conclude that there exists an equilibrium with no information acquisition and no disclosure for all sufficiently large $\kappa$.

If there exists an equilibrium with evidence acquisition, the conjectured non-certified value will be some $\alpha \in \left[0,\mu\right)$. Suppose first that $\mu < \alpha + \gamma$. The line tangent to the curve $\alpha - \kappa c\left(x\right)$ at $x = \mu$ is
\[\kappa c'\left(\mu\right) \left(\mu-x\right) + \alpha \text{.}\]
This line intersects the curve $x - \gamma - \kappa c\left(x\right)$ on the interval $\alpha + \gamma \leq x \leq 1$ if and only if the function \[\kappa c'\left(\mu\right) \left(\mu-x\right) + \alpha - x + \gamma + \kappa c\left(x\right)\]
has a root in $\left[\alpha + \gamma, 1\right]$. The remainder is analogous to the previous paragraph: if $\kappa$ is sufficiently high, the sender prefers to deviate by learning nothing and not disclosing.

Finally, suppose $\mu \geq \alpha + \gamma$. The line tangent to the curve $x - \gamma - \kappa c\left(x\right)$ at $x = \mu$ is 
\[\left(1-\kappa c'\left(\mu\right)\right) \left(x-\mu\right) + \mu - \gamma \text{.}\]
This line intersects the curve $\alpha - \kappa c\left(x\right)$ on the interval $\left[0, \alpha+\gamma\right]$ if and only if the function
\[\psi\left(x\right) \coloneqq \left(1-\kappa c'\left(\mu\right)\right) \left(x-\mu\right) + \mu - \gamma -\alpha + \kappa c\left(x\right) \text{,}\]
has a root in $\left[0,\alpha+\gamma\right]$. The derivative of $\psi$,
\[\psi'\left(x\right) = 1 - \kappa\left(c'\left(\mu\right) - c'\left(x\right)\right) \text{,}\]
is negative for all $\kappa$ sufficiently large. Consequently, for all $\kappa$ sufficiently large, there is a profitable deviation (if $\alpha > 0$) and an equilibrium in which the sender does not acquire any information but gets it certified.
\end{proof}

\subsection{Proposition \ref{costofinfoprop2} Proof}
\begin{proof}
Since Theorem \ref{maintheorem} states that an equilibrium must exist, it suffices to show that for all $\kappa$ sufficiently small there do not exist equilibria in which the sender acquires no information. 

First, suppose that in the purported equilibrium, there is no certification. In this case the conjectured non-certified value is $\mu$. It suffices to show that when $\kappa$ is sufficiently small, the function $\tau$, defined in the previous proposition's proof, has a root in $\left[\mu + \gamma, 1\right]$. Setting $\kappa = 0$, $\tau\left(\mu+\gamma\right) = 0$. Moreover, $\tau'\left(\mu+\gamma\right) = -1$ when $\kappa = 0$. Thus, there is a root, and hence the sender wants to deviate.

Second, suppose that in the purported equilibrium there is certification with probability $1$. In this case $\alpha = 0$. If $\mu \leq \gamma$, we immediately see that this equilibrium is impossible. Thus, let $\mu > \gamma$. Again, it suffices to show that when $\kappa$ is sufficiently small, the function $\psi$, defined in the previous proposition's proof, has a root in $\left[0,\gamma\right]$. Setting $\kappa = 0$, $\psi\left(\gamma\right) = 0$. Moreover, $\psi'\left(\mu+\gamma\right) = 1$ when $\kappa = 0$. Thus, there is a root, so the sender wants to deviate.
\end{proof}
\subsection{Proposition \ref{nodisccostprop} Proof}
\begin{proof}
First, suppose for the sake of contradiction that there is an equilibrium in which the sender acquires no information and does not get it certified with positive probability. In this case, the receiver's belief upon seeing non-certification is precisely $\mu$ and the kink in the sender's value function is at $\mu$, which means that it is never optimal to acquire no information. Second, suppose (also FSOC) that there is an equilbrium in which the sender acquires information. There is obviously a deviation in which the sender acquires no evidence and then discloses it with probability $1$, yielding $S$ her maximal payoff (since she is no longer wasting any resources on evidence acquisition). Finally, there is obviously an equilibrium in which the sender acquires no evidence and gets it certified, since the corresponding value function in her information acquisition problem is strictly concave on $\left[0,1\right]$. \end{proof}

\subsection{\texorpdfstring{$\S$}{text} \ref{uniformsec} Example (Including Proposition \ref{upqc}) Derivation}\label{app11}
A direct computation reveals that given a conjectured lower point of support $\alpha \in \left[0,\mu\right)$, if the optimal information acquisition is bi-pooling (refer to \cite{bipool} or \cite{kleiner2021extreme}) the support of the distribution is $\left\{\alpha + \gamma - \frac{1}{4\kappa}, \alpha + \gamma + \frac{1}{4\kappa}\right\}$. Evidently, $\alpha = \alpha + \gamma - \frac{1}{4\kappa}$ if and only if $\gamma = \frac{1}{4\kappa}$.

Let us begin with this knife-edge case. The Rothschild-Stiglitz representation of the prior (\cite{gent}) is $\int_{0}^{x}udu$ (on $\left[0,1\right]$), whereas the Rotschild-Stiglitz representation of a potential bi-pooling equilibrium is 
\[T\left(x\right) = \begin{cases}
0, \quad &\text{if} \quad 0 \leq x < x_L\\
\left(1 - \kappa \left(1-2\left(x_L\right)\right) \right) \left(x-x_L\right), \quad &\text{if} \quad x_L \leq x < x_L+\frac{1}{2\kappa}\\
x-\frac{1}{2}, \quad &\text{if} \quad x_L+\frac{1}{2\kappa} \leq x\\
\end{cases}\text{,}\]
where $x_L \in \left(0,\mu\right)$.
\begin{lemma}
There exists such bi-pooling equilibria if and only if $\kappa > 1$.
\end{lemma}
\begin{proof}
The line
\[\left(1 - \kappa \left(1-2\left(x_L\right)\right) \right) \left(x-x_L\right)\text{,}\]
is strictly decreasing in $\kappa$ and hence lies (pointwise) above the line $2x_L\left(x-x_L\right)$ when $\kappa \leq 1$ and strictly below this line when $\kappa > 1$. 

The line $2x_L\left(x-x_L\right)$ intersects the curve $x^2/2$ at $x = 2x_L$. Observe that this is (the unique) point of tangency of the line and the curve. We can always find an $x_L$, moreover, such that $x_L + \frac{1}{2\kappa} \leq 1$ and such that the slope of $T$ is positive when $\kappa > 1$. Accordingly, for $\kappa > 1$, there exists such a bi-pooling equilibrium. When $\kappa \geq 1$, $T\left(x\right)$ always intersects $\frac{x^2}{2}$ at some $x \leq 2x_L \leq x_L + \frac{1}{2\kappa}$, and so bi-pooling equilibria do not exist.
\end{proof}
Now let us identify when there exist no-learning equilibria. 
\begin{lemma}
There exists a no-learning equilibrium in which the sender does not acquire evidence and does not disclose if and only if $\kappa \leq 1$ and $\gamma \geq \frac{1}{2} - \frac{\kappa}{4}$ or $\kappa \geq 1$ and $\gamma \geq \frac{1}{4\kappa}$.

There exists a no-learning equilibrium in which the sender does not acquire evidence and always discloses if and only if $\kappa \leq 1$ and $\gamma \leq \frac{\kappa}{4}$ or $\kappa \geq 1$ and $\gamma \leq \frac{1}{2}-\frac{1}{4\kappa}$.
\end{lemma}
\begin{proof}
There are two possibilities: $\alpha = 0$ or $\alpha = \frac{1}{2}$. Let us start with the latter: this is an equilibrium if and only if the line tangent to $\frac{1}{2} - \kappa \left(x-\frac{1}{2}\right)^2$ at $\frac{1}{2}$ lies above the curve $\rho\left(x\right) \coloneqq x-\gamma - \kappa\left(x-\frac{1}{2}\right)^2$ on the interval $\left[\frac{1}{2}+\gamma, 1\right]$. As noted in the text, this holds whenever $\gamma \geq \frac{1}{2}$. Let $\gamma < \frac{1}{2}$. The tangent line is $y = \frac{1}{2}$, and $\rho\left(x\right)$ is maximized at $\frac{1}{2} + \frac{1}{2\kappa}$ if and only if 
\[\frac{1}{2} + \gamma \leq \frac{1}{2} + \frac{1}{2\kappa} \leq 1 \text{,}\]
or else at $\frac{1}{2} + \gamma$ ($\gamma \geq \frac{1}{2\kappa}$) or $1$ ($\frac{1}{2\kappa} \geq \frac{1}{2}$, i.e. $\kappa \leq 1$).

Evidently, if $\gamma \geq \frac{1}{2\kappa}$, the tangent line always lies above $\rho$ and so there exists a no-learning equilibrium of this form. If $\kappa \leq 1$ then $\rho$ is maximized at $x = 1$, yielding value $1 - \gamma - \frac{\kappa}{4}$. $\frac{1}{2}$ is larger than this if and only if $\gamma \geq \frac{1}{2}-\frac{\kappa}{4}$. If $\kappa \geq 1$ and $\gamma \leq \frac{1}{2\kappa}$ the maximal value of $\rho\left(x\right)$ is $\frac{1}{2} + \frac{1}{4\kappa} - \gamma$. Accordingly, $\frac{1}{2}$ lies above this (and so there exists a no-learning equilibrium of this form) provided $\gamma \geq \frac{1}{4\kappa}$.

Now we turn our attention to the other variety of no-learning equilibria. In this case, the analysis is analogous \textit{mutatis mutandis}: we find when the line tangent to $x - \gamma - \kappa\left(x-\frac{1}{2}\right)^2$ at $\frac{1}{2}$ lies above $\phi\left(x\right) \coloneqq - \kappa\left(x-\frac{1}{2}\right)^2$ for all $x \in \left[0,\gamma\right]$.
\end{proof}
Finally, we identify the equilibria (other than the non-generic bi-pooling one when $\gamma = \frac{1}{4\kappa}$) in which the sender acquires evidence.

\paragraph{Proof of Proposition \ref{upqc}}
\begin{proof}
Let $\gamma \neq \frac{1}{4\kappa}$. As noted above, the equilibrium cannot be bi-pooling. Consequently, we solve
\[\max_{t \in \left[0,1\right]}\left\{t\left(\alpha - \kappa \left(\frac{t-1}{2}\right)^2\right) + \left(1-t\right)\left(\frac{t+1}{2} - \gamma - \kappa \left(\frac{t}{2}\right)^2\right)\right\} \text{.}\]
The objective is concave in $t$ if and only if $\kappa \leq 2$. Thus, a necessary condition for information acquisition is $\kappa \leq 2$. Let us now assume this. Taking the FOC and substituting in $t = 2 \alpha$ yields
\[\alpha^{*} = \frac{\kappa - 4 \gamma}{4\left(\kappa - 1\right)} \text{.}\]
First let $\kappa \leq 1$. In this case, $\alpha^{*} \in \left[0,\frac{1}{2}\right]$ if and only if 
$\gamma \in \left[\frac{\kappa}{4}, \frac{1}{2} - \frac{\kappa}{4}\right]$. It is simple to check that the best response to $\alpha^{*}$ is the monotone partitional signal with the truncation point $2\alpha^{*}$ as required. When $\kappa \in \left(1,2\right]$, the best response to the specified $\alpha^{*}$ is a bi-pooling distribution (and not a truncation) and so we conclude that $\kappa \leq 1$ is necessary for this equilibrium with information acquisition to exist. \end{proof}

\subsection{Lemma \ref{lemma61} Proof}
\begin{proof}
Fix $y \in [0, \gamma]$. Using Lemma 1 in \cite{bipool}, we can find $z \in [\mu+\gamma,1]$ such that there does not exist a MPC of $F$ whose support is $\left\{y, z\right\}$. Next we claim that there exists $\kappa_y > 0$ such that $\kappa < \kappa_y$ implies that the lower point in the support of any distribution of posterior means that is not monotone partitional, denote by $x_L$, must satisfy $x_L < y$. To see this, note that because $y \in [0, \gamma]$, $V(y) = \alpha - \kappa c\left(y\right)$; and hence the line tangent to the curve $\alpha - \kappa c\left(x\right)$ at $y$ is
\[\kappa c'\left(y\right) \left(y-x\right) + \alpha.\]
This line intersects the curve $x - \gamma - \kappa c\left(x\right)$ on $[\alpha + \gamma, 1]$ if and only if the function
\[\varphi(x) \coloneqq \kappa c'\left(y\right) \left(y-x\right) + \alpha - x + \gamma + \kappa c\left(x\right)\]
has a root in $[\alpha + \gamma, 1]$. But this must be true since $\varphi\left(\alpha + \gamma\right) = 0$ and $\varphi'(\alpha + \gamma)$. Consequently, $x_L < y$. Analogously, there exists $\kappa_z > 0$ such that $\kappa < \kappa_z$ implies that the lower point in the support of any distribution of posterior means that is not monotone partitional, denote by $x_H$, must satisfy $x_H < y$. Then for $\kappa < \min\{\kappa_y, \kappa_z\}$, Lemma 1 in \cite{bipool} indicates that there does not exist a MPC of $F$ whose support is $\left\{x_L, x_H\right\}$. Thus, for all $\kappa$ sufficiently small, there cannot be any non-monotone partitional equilibrium. \end{proof}

\end{document}